\documentclass[letterpaper,11pt]{article}
\usepackage{xspace}
\usepackage{amsmath}
\newtheorem{thm}{Theorem}[section]
\newtheorem{cor}[thm]{Corollary}
\newtheorem{lemma}[thm]{Lemma}

\newtheorem{claim}[thm]{Claim}
\newtheorem{defn}[thm]{Definition}
\newtheorem{obs}[thm]{Observation}

\newtheorem{definition}{Definition}
\def\squarebox#1{\hbox to #1{\hfill\vbox to #1{\vfill}}}
\newcommand{\qed}{\hspace*{\fill}
\vbox{\hrule\hbox{\vrule\squarebox{.667em}\vrule}\hrule}\smallskip}
\newenvironment{proof}{\noindent{\bf Proof:~~}}{\(\qed\)}

\begin{document}
\setcounter{page}{0}
\author{Elchanan Mossel\footnote{Statistics and Computer Science,
U.C. Berkeley, and Mathematics and Computer Science Weizmann
Institute. Supported by Sloan fellowship in Mathematics, NSF Career
award DMS 0548249, DOD grant N0014-07-1-05-06, and by ISF.
mossel@stat.berkeley.edu}\and Christos
Papadimitriou\footnote{Computer Science Division University of
California at Berkeley, CA, 94720 USA. christos@cs.berkeley.edu}
\and Michael Schapira\footnote{Department of Computer Science, Yale
University, CT, USA, and Computer Science Division, University of
California at Berkeley, CA, USA. Supported by NSF grant 0331548.
michael.schapira@yale.edu.}\and Yaron Singer\footnote{ Computer
Science Division University of California at Berkeley, CA, 94720
USA. yaron@cs.berkeley.edu.}}

\title{VC {\em v.} VCG: Inapproximability of Combinatorial Auctions via Generalizations of the VC Dimension}
\date{}
\maketitle

\begin{abstract}
The existence of incentive-compatible computationally-efficient
protocols for combinatorial auctions with decent approximation
ratios is the paradigmatic problem in computational mechanism
design. It is believed that in many cases good approximations for
combinatorial auctions may be unattainable due to an inherent clash
between truthfulness and computational efficiency. However, to date,
\emph{researchers lack the machinery to prove such results}. In this
paper, we present a new approach that we believe holds great promise
for making progress on this important problem. We take the first
steps towards the development of new technologies for lower bounding
the \emph{VC-dimension of $k$-tuples of disjoint sets}. We apply
this machinery to prove the first \emph{computational-complexity}
inapproximability results for incentive-compatible mechanisms for
combinatorial auctions. These results hold for the important class
of VCG-based mechanisms, and are based on the complexity assumption
that NP has no polynomial-size circuits.
\end{abstract}

\newpage

\section{Introduction}

The field of \emph{algorithmic mechanism design}~\cite{NR01} is
about the reconciliation of bounded computational resources and
strategic interaction between selfish participants. In
\emph{combinatorial auctions}, the paradigmatic problem in this
area, a set of items is sold to bidders with \emph{private}
preferences over \emph{subsets} of the items, with the intent of
maximizing the \emph{social welfare} (\emph{i.e.}, the sum of
bidders' utilities from their allocated items). Researchers
constantly seek auction protocols that are both
\emph{incentive-compatible} and \emph{computationally-efficient},
and guarantee decent approximation ratios. Sadly, to date, huge gaps
exist between the state of the art approximation ratios obtained by
unrestricted, and by truthful, algorithms. It is believed that this
could be due to an inherent clash between the truthfulness and
computational-efficiency requirements, that manifests itself in
greatly degraded algorithm performance. Such tension between the two
desiderata was recently shown to exist in~\cite{PSS08} for a
different mechanism design problem called \emph{combinatorial public
projects}~\cite{SS08}. However, in the context of combinatorial
auctions, due to their unique combinatorial structure, \emph{the
algorithmic game theory community currently lacks the machinery to
prove this}~\cite{T08}.

The celebrated class of {\em Vickrey-Clarke-Groves (VCG)
mechanisms}~\cite{Vic61, Cla71, Gro73} is the only known universal
technique for the design of \emph{deterministic}
incentive-compatible mechanisms (in certain interesting cases VCG
mechanisms are the \emph{only} truthful
mechanisms~\cite{LMNB,DS08,LMN,R79,PSS08}). While a naive
application of VCG is often computationally intractable, more clever
uses of VCG are the key to the best known (deterministic)
approximation ratios for combinatorial auctions~\cite{DNS05,HKMT04}.
For these reasons, the exploration of the computational limitations
of such mechanisms is an important research agenda (pursued
in~\cite{DN07,DS08,LMN,NR00,PSS08}). Recently, it was
shown~\cite{PSS08} that the computational-complexity of VCG-based
mechanisms is closely related to the notion of VC
dimension.~\cite{PSS08} was able to make use of \emph{existing} VC
machinery to prove computational hardness results for combinatorial
public projects. However, for combinatorial auctions, these
techniques are no longer applicable. This is because, unlike
combinatorial public projects, the space of outcomes in
combinatorial auctions does not consist of subsets of the universe
of items, bur rather of \emph{partitions} of this universe (between
the bidders). This calls for the development of new VC machinery for
the handling of such problems.

We formally define the notion of the \emph{VC dimension of
collections of partitions of a universe}\footnote{By a {\em
partition} in this paper we mean an ordered $k$-tuple of disjoint
subsets which, however, may not exhaust the universe.} and present
techniques for establishing lower bounds on this VC dimension. We
show how these results can be used to prove computational-complexity
inapproximability results for VCG-based mechanisms for combinatorial
auctions. (We note that these results actually hold for the more
general class of mechanisms that are affine
maximizers~\cite{LMN,R79}.) Our inapproximability results depend on
the computational assumption that SAT does not have polynomial-size
circuits. Informally, our method of lower bounding\footnote{We use
the term lower bound as a general reference to an inapproximability
result. Hence, a lower bound of $\frac{1}{2}$ means (as we are
looking at a maximization problem) that no approximation better than
$\frac{1}{2}$ is possible. This use is similar to that of Hastad
in~\cite{Has01}.} the approximability of VCG-based mechanisms via VC
arguments is the following: We consider well-known auction
environments for which \emph{exact} optimization is NP-hard. We show
that if a VCG-based mechanism \emph{approximates closely} the
optimal social welfare, then it is implicitly solving {\em
optimally} a smaller, but still relatively large, optimization
problem of the same nature
--- an NP-hard feat. We establish this by showing that the subset
of outcomes (\emph{partitions} of items) considered by the VCG-based
mechanism ``shatters'' a relatively large subset of the items.

Our results imply the first \emph{computational complexity} lower
bounds for VCG-based mechanisms, and truthful mechanisms in general,
for combinatorial auctions (with the possible exception of a result
in~\cite{LMN} for a related auction environment). In fact, we show
that in some auction environments these results actually apply to
\emph{all} truthful algorithms. We illustrate our techniques via
$2$-bidder combinatorial auctions, and believe that our approach
holds great promise for making progress on the general problem. It
is also our belief that the notion of the VC dimension of $k$-tuples
of disjoint sets is of independent interest, and suggests many new
and exciting problems in combinatorics.

\subsection{Related Work}

Combinatorial auctions have been extensively studied in both the
economics and the computer science literature~\cite{BN,CSS05,VV03}.
It is known that if the preferences of the bidders are unrestricted
then no constant approximation ratios are achievable (in polynomial
time)~\cite{LOS:J,NS06}. Hence, much research has been devoted to
the exploration of restrictions on bidders' preferences that allow
for good approximations, \emph{e.g.}, for \emph{subadditive}, and
\emph{submodular}, preferences \emph{constant} approximation ratios
have been obtained~\cite{DNS05,DS06,F06,FV06, LLN06,V08}. In
contrast, the known \emph{truthful} approximation algorithms for
these classes have \emph{non-constant} approximation
ratios~\cite{D07,DNS05,DNS06}. It is believed that this gap may be
due to the computational burden imposed by the truthfulness
requirement. However, to date, this belief remains unproven. In
particular, no \emph{computational complexity} lower bounds for
truthful mechanisms for combinatorial auctions are known.

Vickrey-Clarke-Groves (VCG) mechanisms~\cite{Vic61, Cla71, Gro73},
named after their three inventors, are the fundamental technique in
mechanism design for inducing truthful behaviour of strategic
agents. Nisan and Ronen~\cite{NR00,NR01} were the first to consider
the computational issues associated with the VCG technique. In
particular,~\cite{NR00} defines the notion of VCG-Based mechanisms.
VCG-based mechanisms have proven to be useful in designing
approximation algorithms for combinatorial
auctions~\cite{DNS05,HKMT04}. In fact, the best known
(deterministic) truthful approximation ratios for combinatorial
auctions were obtained via VCG-based mechanisms~\cite{DNS05,HKMT04}
(with the notable exception of an algorithm in~\cite{BGN03} for the
case that many duplicates of each item exist). Moreover, Lavi,
Mu'alem and Nisan~\cite{LMN} have shown that in certain interesting
cases VCG-based mechanisms are essentially the only truthful
mechanisms (see also~\cite{DS08}).

Dobzinski and Nisan~\cite{DN07} tackled the problem of proving
inapproximability results for VCG-based mechanisms by taking a
\emph{communication complexity}~\cite{Y79, KN97} approach. Hence, in
the settings considered in~\cite{DN07}, it is assumed that each
bidder has an \emph{exponentially large string of preferences} (in
the number of items). However, real-life considerations render
problematic the assumption that bidders' preferences are exponential
in size. Our intractability results deal with bidder preferences
that are \emph{succinctly described}, and therefore relate to
\emph{computational complexity}. Thus, our techniques enable us to
prove lower bounds even for the important case in which bidders'
preferences can be concisely represented.

The connection between the VC dimension and VCG-based mechanisms was
observed in~\cite{PSS08}, where a general (i.e., not restricted to
VCG-based mechanisms) inapproximability result was presented, albeit
in the context of a different mechanism design problem, called
\emph{combinatorial public projects} (see also~\cite{SS08}). The
analysis in~\cite{PSS08} was carried out within the standard VC
framework, and so it relied on \emph{existing} machinery (namely,
the Sauer-Shelah Lemma~\cite{Sauer72, Shelah72} and its
probabilistic version due to Ajtai~\cite{Ajtai98}). To handle the
unique technical challenges posed by combinatorial auctions
(specifically, the fact that the universe of items is
\emph{partitioned} between the bidders) \emph{new} machinery is
required. Indeed, our technique can be interpreted as \emph{an
extension of the Sauer-Shelah Lemma to the case of partitions}
(Lemma~\ref{lem-VC-partitions} in Sec.~\ref{sec_partitions-vc}).

The VC framework has received much attention in past decades (see,
\emph{e.g.},~\cite{ABCH97,BartlettL98,MendelsonV02} and references
therein), and many generalizations of the VC dimension have been
proposed and studied (\emph{e.g.},~\cite{Alon83}). To the best of
our knowledge, none of these generalizations captures the case of
$k$-tuples of disjoint subsets of a universe considered in this
paper. In addition, no connection was previously made between the
the VC dimension and the approximability of combinatorial auctions.

\subsection{Organization of the Paper}

In Sec.~\ref{sec_partitions-vc} we present our approach to analyzing
the VC dimension of partitions. In Sec.~\ref{sec_implications} we
prove our inapproximability results for combinatorial auctions. We
conclude and present open questions in Sec.~\ref{sec_open}.

%
%In Section~\ref{sec-vc-vcg} we discuss the outline of our approach
%to proving inapproximability results for maximal-in-range
%mechanisms. In particular, we present the interesting connections
%between our goal and the classic notion of VC dimension. In
%Section~\ref{sec-first} we present our lower bound for the special
%case of the $2$-bidder auction described above in which all items
%must be allocated. In Section~\ref{sec-main} we present our result
%for the case in which this assumption is removed. Finally, in
%Section~\ref{sec_open} we present open questions and directions for
%future research.

\section{The VC Dimension of Partitions}\label{sec_partitions-vc}

The main hurdle in combinatorial auctions stems from the fact that
the outcomes are \emph{partitions} of the set of items.
Approximation algorithms for combinatorial auctions can be thought
of as functions that map bidders' preferences to partitions of
items. This motivates our extension of the standard notion of VC
dimension to \emph{the VC dimension of partitions}. This section
presents this notion of VC dimension and lower bounding techniques.
In Sec.~\ref{sec_implications} we harness this machinery to prove
computational complexity lower bounds for combinatorial auctions.

\subsection{The VC Dimension and $\alpha$-Approximate Collections of
Partitions}\label{subsec_approx-partitions}

We focus on partitions that consist of two disjoint subsets (our
definitions can easily be extended to $k$-tuples). Our formal
definition of a partition of a universe is the following:

\begin{defn} [partitions]\label{def_n-partition}
A partition $T=(T_1,T_2)$ of a universe $U=\{1,...,m\}$ is a pair of
two disjoint subsets of $U$, i.e. $T_1,T_2 \subseteq [m]$ and $T_1
\cap T_2 = \emptyset$.
\end{defn}

Observe that we \emph{do not} require that every element in the
universe appear in one of the two disjoint subsets that form a
partition. This definition of partitions will later enable us to
address crucial aspects of combinatorial auctions. We refer to
partitions that \emph{do} exhaust the universe (\emph{i.e.}, cover
all elements in the universe) as ``\emph{covering partitions}'' (we
shall refer to not-necessarily-covering partitions as
``\emph{general partitions}'').

\begin{defn} [covering partitions]\label{def_n-partition-cover}
A partition $(T_1,T_2)$ of a universe $U$ is said to cover $U$ if
$T_1 \cup T_2=U$. $C(U)$ is defined to be the set of all partitions
that cover $U$.
\end{defn}

For every subset $E$ of a universe $U$, we can define (in an
analogous way) what a partition of $E$ is, and denote by $P(E)$ the
set of all partitions of $E$ and by $C(E)$ the set of all partitions
of $E$ that cover $E$.

\begin{defn} [projections]\label{def_projection}
The projection of a partition $(S_1,S_2)\in P(U)$ on $E\subseteq U$,
denoted by $(S_1,S_2)_{|E}$, is the partition $(S_1\cap E,S_2\cap
E)\in P(E)$. For any collection of partitions $R\subseteq P(U)$ we
define $R$'s projection on $E\subseteq U$, $R_{|E}$, to be
$R_{|E}=\{(T_1,T_2)|\ \exists (S_1,S_2)\in R\ s.t.\
(S_1,S_2)_{|E}=(T_1,T_2)\}$.
\end{defn}

Observe that if a partition $(S_1,S_2)$ of $E\subseteq U$ is in
$C(E)$, then for any $E'\subseteq E$ $(S_1,S_2)_{|E'}\in C(E')$. We
are now ready to define the notions of shattering and VC dimension
in our context:

\begin{defn} [shattering]\label{def_shattered}
A subset $E\subseteq U$ is said to be shattered by a collection of
partitions $R\subseteq P(U)$ if $C(E)\subseteq R_{|E}$.
\end{defn}

Observe that that if $E\subseteq U$ is shattered by a collection of
partitions $R\subseteq P(U)$ then so are all subsets of $E$. By
Definition~\ref{def_shattered}, for a subset $E\subseteq U$ to be
shattered it suffices that $C(E)\subseteq R_{|E}$. We do not require
that $R_{|E}=P(E)$. However, we note that all of our results for
general partitions actually also hold for the latter (stronger)
requirement.

\begin{defn} [VC dimension]\label{def_VC}
The VC dimension $VC(R)$ of a collection of partitions $R\subseteq
P(U)$ is the cardinality of the biggest subset $E\subseteq U$ that
is shattered by $R$.
\end{defn}

We now introduce the useful concept of $\alpha$-approximate
collections of partitions. Informally, a collection $R$ of
partitions is $\alpha$-approximate if, for every partition $S$ of
the universe (not necessarily in $R$), there is some partition in
$R$ that is ``not far'' (in terms of $\alpha$) from $S$. We are
interested in the connection between the value of $\alpha$ of an
$\alpha$-approximate collection of partitions and its VC dimension.
This will play a major role in the proofs of our results for
combinatorial auctions.

\begin{definition} [$\alpha$-approximate collections of
partitions]\label{def_approximate-partitions} Let $R$ be a
collection of partitions of a universe $U$. $R$ is said to be
$\alpha$-approximate if for every partition $S=(S_1,S_2)\in P(U)$
there exists some partition $T=(T_1,T_2)\in R$ such that $|S_1\cap
T_1|+|S_2\cap T_2|\geq \alpha(|S_1|+|S_2|)$.
\end{definition}

\subsection{Lower Bounding the VC Dimension of Collections of
Covering Partitions}

When dealing with collections of covering partitions it is possible
to use existing VC machinery to lower bound their VC dimension.
Specifically, a straightforward application of the Sauer-Shelah
Lemma~\cite{Sauer72,Shelah72} implies that:

\begin{lemma} [lower bounding the VC dimension of covering partitions]\label{lem_vc_cover}
For every $R\subseteq C(U)$ it holds that $VC(R)=\Omega({\log
|R|\over \log |U|})$.
\end{lemma}
\begin{proof}
Let $R_1=\{S_1|\ \exists S_2\ s.t.\ (S_1,S_2)\in R\}$. Because $R$
only consists of covering partitions it must be that $|R_1|=|R|$. We
now recall the Sauer-Shelah Lemma:

\begin{lemma} (\cite{Sauer72,Shelah72})
For any family $Z$ of subsets of a universe $U$, there is a subset
$E$ of $U$ of size $\Theta({\log |Z|\over \log |U|})$ such that for
each $E'\subseteq E$ there is a $Z'\in Z$ such that $E'=Z'\cap E$.
\end{lemma}

The Sauer-Shelah Lemma, when applied to $R_1$ implies the existence
of a set $E$ (as in the statement of the lemma) of size
$\Omega({\log |R_1|\over \log |U|})$ (that is, a large set that is
shattered in the traditional sense). The fact that all partitions in
$R$ are covering partitions now immediately implies that
$VC(R)=\Omega({\log |R_1|\over \log |U|})=\Omega({\log |R|\over \log
|U|})$.
\end{proof}

Lemma~\ref{lem_vc_cover} enables us to prove a lower bounds on the
VC dimension of $(\frac{1}{2}+\epsilon)$-approximate collections of
\emph{covering} partitions.

\begin{thm}\label{thm_vc_covering}
Let $R\subseteq C(U)$. If $R$ is
$(\frac{1}{2}+\epsilon)$-approximate (for any small constant
$\epsilon>0$) then there exists some constant $\alpha>0$ such that
$VC(R)\geq m^{\alpha}$.
\end{thm}

\begin{proof} [sketch] We use a probabilistic construction argument:
Consider a partition of $U$, $S=(S_1,S_2)$, that is chosen,
uniformly at random, out of all possible partitions in $C(U)$.
Observe that (by the Chernoff bounds), for every partition
$T=(T_1,T_2)\in R$, the probability that $|S_1\cap T_1|+|S_2\cap
T_2|\geq
(\frac{1}{2}+\epsilon)(|S_1|+|S_2|)=(\frac{1}{2}+\epsilon)m$ is
exponentially small in $m$. Hence, for $R$ to be
$(\frac{1}{2}+\epsilon)$-approximate it must contain exponentially
many partitions in $C(U)$. We can now apply Lemma~\ref{lem_vc_cover}
to conclude the proof.
\end{proof}

\subsection{Lower Bounding the VC Dimension of Collections of
General Partitions}

Observe that the proof Lemma~\ref{lem_vc_cover} heavily relied on
the fact that the partitions considered were covering partitions.
Dealing with collections of general partitions necessitates the
development of different techniques for lower bounding the VC
dimension. We shall now present such a method, that can be regarded
as an extension of the Sauer-Shelah Lemma~\cite{Sauer72,Shelah72} to
the case of collections of general partitions:

\begin{defn} [distance between partitions]
Given a universe $U$, two partitions in $P(U)$, $(T_1,T_2)$ and
$(T'_1,T'_2)$, are said to be $b$-far (or at distance $b$) if $|T_1
\cap T'_2| + |T'_1 \cap T_2|\geq b$.
\end{defn}

\begin{defn}
Let $t(\epsilon,k,m)$ be the smallest possible number of sets $E
\subset [m]$ that are shattered by a set $R$ of partitions of size
$k$, such that every two elements in $R$ are at least $\epsilon
m$-far.
\end{defn}

% EM: Changed here
\begin{obs}\label{obs_vc-counting}
Suppose $k \geq 1$ and $\epsilon m \geq 1$. Then if $t(\epsilon,k,m)
> \sum_{i=1}^r {m \choose r}$ then the VC dimension of any
collection of partitions of size at least $k$ for which every two
partitions are at least $\epsilon m$-far has to be at least $r+1$.
%(this is derived from
%observation~\ref{obs_shattered-subsets} and a simple counting
%argument).
\end{obs}
\begin{proof}
The proof follows from the fact that $t(\epsilon,k,m)\geq
\sum_{i=0}^r {m \choose r}$ is a bound on the number of sets of size
at most $r$.
\end{proof}

\begin{lemma} [lower bounding the VC dimension of general partitions]\label{lem-VC-partitions}
For all $\epsilon>0,k,m$, $t(\epsilon,k,m) \geq k^{\alpha}$ for some
constant $\alpha>0$.
\end{lemma}

The proof follows the basic idea
of~\cite{ABCH97,BartlettL98,MendelsonV02}. Our novel observation is
that the same proof strategy applies with our new definition of
distance.

% EM: End changes

\begin{proof} [Sketch] Fix $\epsilon>0,k,m$. We wish to prove that
$t(\epsilon,k,m) \geq k^{\alpha}$, for some constant $\alpha>0$. We
shall bound $t(\epsilon,k,m)$ by induction ($\epsilon$ shall remain
fixed throughout the proof and the induction is on $k$ and $m$). Let
$R$ be some collection of partitions as in the statement of the
lemma. Arbitrarily partition $R$ into pairs. Since the partitions
that make up each pair are at least $\epsilon m$-far there must
exist (via simple counting) an element $e\in U$, such that in at
least $\frac{\epsilon k}{2}$ pairs $(T_1,T_2),(T'_1,T'_2)$, $e\in
T_1\cap T'_2$ or $e\in T'_1\cap T_2$. Let $R'\subseteq R$ be the
collection of all partitions $(T_1,T_2)$ in $R$ in which $e\in T_1$.
Let $R''\subseteq R$ be the collection of all partitions $(T_1,T_2)$
in $R$ in which $e\in T_2$. By the arguments above we are guaranteed
that $|R'|\geq \frac{\epsilon k}{2}$ and $|R''|\geq \frac{\epsilon
k}{2}$.

Let $I$ be all the subsets of $U$ that are shattered by $R$. We wish
to lower bound $|I|$. Let $R'_{-e}$ be all the partitions of
$U\setminus \{e\}$ we get by removing $e$ from $T_1$ for every
partition $(T_1,T_2)\in R'$. Let $I'$ be all the subsets of
$U\setminus\{e\}$ shattered by $R'_{-e}$. As there are at least
$\frac{\epsilon k}{2}$ sets in $R'$, by definition $|I'|\geq
t(\epsilon,\frac{\epsilon k}{2}, m-1)$. Similarly, let $R''_{-e}$ be
all the partitions of $U\setminus \{e\}$ we get by removing $e$ from
$T_1$ for every partition $(T_1,T_2)\in R''$. Let $I''$ be all the
subsets of $U\setminus\{e\}$ shattered by $R''_{-e}$. As there are
at least $\frac{\epsilon k}{2}$ sets in $R''$, by definition
$|I''|\geq t(\epsilon,\frac{\epsilon k}{2}, m-1)$.

We claim that $|I|\geq |I'|+|I''|$. To see why this is true consider
the following argument: All sets in $I'\setminus I''$ and in
$I''\setminus I'$ are distinct and belong to $I$. Let $S$ be a set
in $I'\cap I''$. Observe that this means that not only is $S$ in
$I$, but so is $S\cup\{e\}$. So, $I\geq |I'\setminus
I''|+|I'\setminus I''|+2|I'\cap I''|=|I'|+|I''|$. Hence,
$t(\epsilon,k,m) \geq 2\times t(\epsilon,\frac{\epsilon k}{2},
m-1)$. We now use the induction hypothesis to conclude the proof.
\end{proof}

Lemma~\ref{lem-VC-partitions} and Observation~\ref{obs_vc-counting}
imply the following important corollary:

% EM: Changed
\begin{cor}\label{cor-VC-partitions}
For every constant $\alpha>0$, and (sufficiently small) constant
$\epsilon
>0$, there exists a $\beta>0$ such that, if $R\subseteq P(U)$ and it
holds that: (1) $|R|\geq e^{m^{\alpha}}$, and (2) every two
partitions in $R$ are $\epsilon m$-far, then $VC(R)\geq m^{\beta}$.
\end{cor}

We shall now discuss the connection between the value of $\alpha$ of
an $\alpha$-approximate collection of partitions and its VC
dimension. We prove the following theorem:

\begin{thm}\label{thm_vc-non-covering}
Let $R\subseteq P(U)$. If $R$ is
$(\frac{3}{4}+\epsilon)$-approximate (for any constant $\epsilon>0$)
then there exists some constant $\alpha>0$ such that $VC(R)\geq
m^{\alpha}$.
\end{thm}

\begin{proof} To prove the theorem we use the following claim:

\begin{claim}\label{claim-prob-hamming} For every
small constant $\delta>0$, there is a family $F$ of partitions
$(T_1,T_2)$ in $C([m])$ and a constant $\alpha>0$ such that
$|F|=e^{\alpha m}$ and every two partitions in $F$ are at least
$\frac{1-\delta}{2}m$-far.
\end{claim}

\begin{proof}
We will prove the claim for partitions $T=(T_1,T_2)$ where $T_1 \cup
T_2 = [m]$. For such partitions, the distance between partition
$T=(T_1,T_2)$ and $T'=(T_1',T_2')$ is just the size of the symmetric
difference of $T_1$ and $T_2$. The existence of the desired
collection now follows from the existence of good codes, see
e.g.~\cite{Sudan04}. For completeness we include the standard
construction to show the existence of $F$.

Let $T=(T_1,T_2)$ and $T'=(T'_1,T'_2)$ be two partitions in $C([m])$
chosen at random in the following way: For each item $j\in [m]$ we
choose, uniformly at random, whether it will be placed in $T_1$ or
in $T_2$. Similarly, we choose, uniformly at random, whether each
item $j$ shall be placed in $T'_1$ or $T'_2$. Using standard
Chernoff arguments it is easy to show that the probability that
there are at least $\frac{m+\delta}{2}$ that appear in either
$T_1\cap T'_1$ or $T_2\cap T'_2$ is exponentially small in
$\epsilon'$. Observe that this immediately implies (by our
definition of distance) that the probability that the distance
between $T$ and $T'$ is less than $\frac{1-\delta}{2}m$ is
exponentially small i $\delta$. Hence, a family $F$ of exponential
size must exist.
\end{proof}

\begin{lemma}\label{lem_MIR-subrange}
Let $\epsilon>0$. Let $R\subseteq P(U)$ such that $R$ is
$(\frac{3}{4}+\epsilon)$-approximate. Then, there is a subset of
$R$, $R'$, of size exponential in $m$ such that every two elements
of $R'$ are at least $\alpha m$-far (for some constant $\alpha>0$).
\end{lemma}

\begin{proof} By Claim~\ref{claim-prob-hamming} we know that, for our universe of
$U$, there exists an exponential-sized family of partitions in
$C(U$), $F$, such that every two partitions in $F$ are at least
$\frac{1-\delta}{2}m$-far (for some arbitrarily small $\delta>0$).
Fix some $T,T'\in F$. By definition of $F$, $T$ and $T'$ are
identical only on at most $\frac{1+\delta}{2}m$ elements (that is,
only for at most $\frac{1+\delta}{2}m$ items $j$, either $j\in
T_1\cap T'_1$ or $j\in T_2\cap T'_2$). Let $R_T$ and $R_T'$
represent two partitions in $R$ that obtain $\frac{3}{4}+\epsilon$
``approximations'' for $T$ and $T$', respectively (because $R$ is
$(\frac{3}{4}+\epsilon)$-approximate such partitions must exist).
Even if we assume that both $R^T$ and $R^{T'}$ are identical on all
elements on which $T$ and $T'$ are identical, we are still left with
$\frac{1-\delta}{2}m$ elements. Observe that for each such element,
if $R^T$ and $R^{T'}$ are identical on it, it holds that it can only
contribute to the approximation obtained by \emph{one} of them. This
implies that to obtain the promised approximation $R^T$ and $R^{T'}$
must differ on quite a lot (a constant fraction) of the elements in
$U$. This, in turn, implies that there is some $\alpha>0$ such that
$R^T$ and $R^{T'}$ are $\alpha m$-far.
\end{proof}

Corollary~\ref{cor-VC-partitions} now concludes the proof of the
theorem.
\end{proof}

\section{Implications For VCG-Based Mechanisms}\label{sec_implications}

In this section we present the connection between our results for
collections of partitions in Sec.~\ref{sec_partitions-vc} and the
problem of social-welfare-maximization in combinatorial auctions. We
use the VC dimension framework developed in the previous section to
present a general technique for proving computational complexity
lower bounds for VCG-based mechanisms.

\subsection{Maximal-In-Range Mechanisms for Combinatorial Auctions}

\paragraph*{2-bidder combinatorial auctions.} We consider auction
environments of the following form: There is a set of items
$1\ldots,m$ that are sold to $2$ bidders, $1$ and $2$. Each bidder
$i$ has a private \emph{valuation function} (sometimes simply
referred to as a \emph{valuation}) $v_i$ that assigns a nonnegative
real value to every subset of the items. $v_i(S)$ can be regarded as
$i$'s maximum willingness to pay for the bundle of items $S$. Each
$v_i$ is assumed to be \emph{nondecreasing}, \emph{i.e.}, $\forall S
\subseteq T$ it holds that $v_i(S)\leq v_i(T)$. The objective is
find a partition of the items $(S_1,S_2)$ between the two bidders
that maximizes the social welfare, \emph{i.e.}, the expression
$\Sigma_i v_i(S_i)$.

It is known that optimizing the social welfare value in $2$-bidder
combinatorial auctions is computationally intractable even for very
restricted classes of valuation functions. In
particular,~\cite{LLN06} shows that this task is NP-hard even for
the simple class of \emph{capped additive valuations}:

\begin{definition} [additive valuations]
A valuation function $a$ is said to be additive if there exist
per-item values $a_{i1},\ldots,a_{im}$, such that for every bundle
$S\subseteq [m]$, $a(S)=\Sigma_{j\in S}\ a_{ij}$.
\end{definition}

\begin{definition} [capped additive valuations]
A valuation function $v$ is said to be a capped additive valuation
if there exist an additive valuation $a$ and a real value $B$, such
that for every bundle $S\subseteq [m]$, $v(S)=\min\{a(S), B\}$.
\end{definition}

Intuitively, a bidder has a capped additive valuation if his value
for each bundle of items is simply the additive sum of his values
for the items in it, up to some maximum amount he is willing to
spend. This class of valuations shall be used throughout this
section to illustrate our impossibility results (as we aim to prove
inapproximablity results the restrictedness of this class works to
our advantage).

\paragraph*{Maximal-in-range mechanisms.} Mechanisms that rely
on the VCG technique to ensure truthfulness (VCG-based mechanisms)
are known to have the useful combinatorial property of being
\emph{maximal-in-range}~\cite{DN07,NR00,R79}:\footnote{Maximal-in-range
mechanisms are a special case of a more general class of mechanisms
called ``affine maximizers''~\cite{R79,LMN}. All of the results in
this paper actually apply to this more general class.}
Maximal-in-range mechanisms are mechanisms that always
\emph{exactly} optimize over a (fixed) set of outcomes. In our
context, this means that for every maximal-in-range mechanism $M$
there exists some $R_M\subseteq P([m])$ such that $M$ always outputs
an optimal outcome \emph{in $R_M$} (with respect to social-welfare
maximization). We refer to $R_M$ as $M$'s \emph{range}.

\begin{definition} [maximal-in-range mechanisms]
A mechanism $M$ is maximal-in-range if there is a collection of
partitions $R_M\subseteq P([m])$ such that for every pair of
valuations, $(v_1,v_2)$, $M$ outputs a partition $(T_1,T_2)\in
argmax_{(S_1,S_2)\in R_M}\ v_1(S_1)+v_2(S_2)$.
\end{definition}

It is know that every maximal-in-range mechanism can be made
incentive compatible via the VCG technique~\cite{NR00,NR01}. This
suggests a general way for the design of truthful mechanisms for
combinatorial auctions: Fix the range $R_M$ of a maximal-in-range
mechanism $M$ to be such that (1) optimizing over $R_M$ can be done
in polynomial time, and (2) the optimal outcome \emph{in $R_M$}
always provides a ``good'' approximation to the
\emph{globally-optimal} outcome. This approach was shown to be
useful in~\cite{DNS05,HKMT04}.

Observe that the maximal-in-range mechanism in which $R_M$ contains
\emph{all} possible partitions of items is computationally
intractable even for capped additive valuations. In contrast, the
fact that bidders' valuations are nondecreasing implies the
following general upper bound:

\begin{obs} [the trivial upper bound]
For any $2$-bidder combinatorial auction, the maximal-in-range
mechanism $M$ for which $R_M=\{([m],\emptyset),(\emptyset,[m])\}$
provides a $\frac{1}{2}$-approximation to the optimal social
welfare.
\end{obs}

That is, the maximal-in-range mechanism that bundles all items
together and allocates them to the bidder with the highest value
provides a $\frac{1}{2}$-approximation to the optimal social
welfare. This mechanism is easy to implement in a
computationally-efficient manner as it only requires learning the
value of each bidder for the bundle of all items.

\paragraph*{Is the trivial upper bound optimal?} Naturally,
we are interested in the question of whether a more clever choice of
range than $\{([m],\emptyset),(\emptyset,[m])\}$ can lead to better
approximation ratios (without jeopardizing computational
efficiency). Let us consider $2$-bidder combinatorial auctions with
capped additive valuations. For this restricted case, a non-truthful
PTAS exists~\cite{AM04}. Can a similar result be obtained via a
maximal-in-range mechanism? We show that the answer to this question
is \emph{No} by proving that the approximation ratios obtained by
computationally-efficient VCG-based mechanisms are always bounded
away from $1$. We stress that these are the first computational
complexity lower bounds on the approximability of VCG-based
mechanisms for combinatorial auctions. In fact, as we shall later
show, in certain cases these bounds extend to \emph{all}
incentive-compatible mechanisms.

\subsection{Putting the VC in VCG}

We now present our method of proving lower bounds on the
approximability of VCG-based mechanisms using the VC framework. On a
high level, our technique for proving that a maximal-in-range
mechanism $M$ cannot obtain an $\alpha$-approximation consists of
three steps:

\begin{itemize}

\item Observe that $M$'s range must be an $\alpha$-approximate collection of partitions.

\item Conclude (from our results in
Sec.~\ref{sec_partitions-vc}) the existence of a shattered set of
items of size $m^{\alpha}$ (if $\alpha$ is sufficiently high).

\item Show a non-uniform reduction from NP-hard $2$-bidder combinatorial
auctions with $m^{\alpha}$ items to the optimization problem solved
by $M$.

\end{itemize}

We illustrate these three steps by proving a lower bound of
$\frac{3}{4}$ for $2$-bidder combinatorial auctions with capped
additive valuations (which naturally extends to the more general
classes of submodular, and subadditive, valuations). We stress that
our proof technique can be applied to prove the \emph{same} lower
bound for practically any NP-hard $2$-bidder combinatorial auction
environment. Essentially, our only requirement from the class of
valuations is that it be expressive enough to contain the class of
0/1-additive valuations defined below.

\begin{definition} [0/1-additive valuations]
A valuation $v$ is said to be 0/1-additive if it is an additive
valuation in which all the per-item values are in $\{0,1\}$.
\end{definition}

We make the following observation:

\begin{obs}\label{obs_connecting}
Any $\alpha$-approximation maximal-in-range mechanism for $2$-bidder
combinatorial auctions with 0/1-additive valuations has a range that
is an $\alpha$-approximate collection of partitions.
\end{obs}
\begin{proof}
A 0/1-additive valuation can be regarded as an indicator function
that specifies some subset of the universe (that contains only the
items that are assigned a value of $1$). Hence, pairs of such
valuations that specify disjoint subsets correspond to partitions of
the universe. Now, it is easy to see that, by definition, the range
of an $\alpha$-approximation maximal-in-range mechanism must be an
$\alpha$-approximate collection of partitions.
\end{proof}

Observation~\ref{obs_connecting} enables us to make use of
Theorem~\ref{thm_vc-non-covering} to conclude that:

\begin{thm}\label{thm_main}
The range of any $(\frac{3}{4}+\epsilon)$-approximation
maximal-in-range mechanism for $2$-bidder combinatorial auctions
with 0/1-additive valuations shatters a set of items of size
$m^{\alpha}$ (for some constant $\alpha>0$).
\end{thm}

We can now exploit the existence of a large shattered set of items
to prove our lower bound by showing a non-uniform reduction from an
NP-hard optimization problem:

\begin{thm}\label{thm_lower_non-cover}
No polynomial-time maximal-in-range mechanism obtains an
approximation ratio of $\frac{3}{4}+\epsilon$ for $2$-bidder
combinatorial auctions with capped additive valuations unless NP
$\subseteq$ P/$poly$.
\end{thm}

\begin{proof}
Let $M$ be a mechanism as in the statement of the theorem. Since
0/1-additive valuations are a special case of capped additive
valuations, by Theorem~\ref{thm_main} there exists a constant
$\alpha>0$ such that $R_m$ shatters a set of items $E$ of size
$m^{\alpha}$. Therefore, given an auction with $m^{\alpha}$ items
and capped additive valuation functions $v_1,v_2$ we can identify
each item in this smaller auction with some unique item in $E$, and
construct valuation functions $v'_1,v'_2$, such that $v'_i$ is
identical to $v_i$ on $E$ and assigns $0$ to all other items.
Observe that this means that $M$ will output for $v'_1,v'_2$ the
optimal solution for $v_1,v_2$ (as $M$'s range contains all
partitions in $C(E)$). We now have a non-uniform reduction from an
NP-hard problem (social-welfare maximization in the smaller auction)
to the optimization problem solved by $M$.
\end{proof}

Recall that the trivial upper bound provides an approximation ratio
of $\frac{1}{2}$. We leave the problem of closing the gap between
this upper bound and our lower bound open. We conjecture that the
trivial upper bound is, in fact, tight. This conjecture is motivated
by the following result.

\paragraph*{The allocate-all-items
case.} We now consider the well-studied case that the auctioneer
must allocate all items~\cite{DS08,LMN}. Observe that, in this case,
the range of a maximal-in-range mechanism can only consist of
covering partitions, for which stronger results are obtained in
Sec.~\ref{sec_partitions-vc}. This enables us to use our technique
to prove the following theorem:

\begin{thm}\label{thm_lower_cover}
For the allocate-all-items case, no polynomial-time maximal-in-range
mechanism obtains an approximation ratio of $2-\epsilon$ for
$2$-bidder combinatorial auctions with capped additive valuations
unless NP $\subseteq$ P/$poly$.
\end{thm}

If bidders have subadditive valuations, and all items are allocated,
then maximal-in-range mechanisms are the \emph{only} truthful
mechanisms~\cite{DS08}. Since capped additive valuations are a
special case of subadditive valuations, the lower bound in
Theorem~\ref{thm_lower_cover} holds for \emph{all} truthful
mechanisms in this more general environment. In
Appendix~\ref{apx_ajtai}, we show that, for a superclass of capped
additive valuations, it is possible to relax the computational
assumption in the statement of Theorem~\ref{thm_lower_cover} to the
assumption that NP is not contained in BPP. This is achieved by
using Ajtai's~\cite{Ajtai98} probabilistic version of the
Sauer-Shelah Lemma.

\section{Discussion and Open Questions}\label{sec_open}

We believe that our work opens a new avenue for proving
complexity-theoretic inapproximability results for maximal-in-range
mechanisms for auctions. In particular, the following important
questions remain wide open:

\begin{enumerate}

\item {\bf Lower bounding the VC dimension of $k$-tuples of disjoint sets, where $k\geq 3$.}
Lemma~\ref{lem-VC-partitions} presents a lower bound on the VC
dimension of pairs of disjoint sets. This enabled us to prove
inapproximability results for $2$-bidder combinatorial auctions. We
believe that the development of advanced VC technology for
$k$-tuples of disjoint sets, where $k\geq 3$, is the key to proving
such results for $k$-bidder combinatorial auctions.

We note that even if all bidders in an $n$-bidder combinatorial
auctions have capped additive valuations, the best (deterministic)
approximation ratio obtained by a truthful mechanism is, to date,
$O(\min\{n,\sqrt{m}\})$ (constant non-truthful approximation ratios
exist). This truthful approximation is achieved by a simple
maximal-in-range mechanism~\cite{DNS05} (using randomization,
improved, but still non-constant, approximation ratios are
achievable~\cite{DNS06, D07}). A straightforward application of our
techniques yields the following result:

\begin{thm}
For any constant number of bidders $n$ with capped additive
valuations, and for any $\epsilon>0$, no maximal-in-range mechanism
can obtain an approximation ratio of $\frac{n+1}{2n}+\epsilon$
unless NP $\subseteq$ P/$poly$.
\end{thm}

We conjecture that a much stronger result is true:

\noindent {\bf Conjecture:}\emph{ No maximal-in-range mechanism can
obtain a constant approximation ratio for the $n$-bidder case.}

\item  {\bf Improved lower bounds for $2$-tuples of disjoint
sets.} We conjecture that, even in $2$-bidder combinatorial auctions
with capped additive valuations, the trivial upper bound of $1\over
2$ is the best possible for maximal-in-range mechanisms (we prove
that this is true under the allocate-all-items assumption). We
believe that such a result can be achieved by strengthening our VC
dimension lower bound in Theorem~\ref{thm_vc-non-covering}.

\noindent {\bf Conjecture:} \emph{No maximal-in-range mechanism can
obtain an approximation ratio of $\frac{1}{2}+\epsilon$ for
$2$-bidder combinatorial auctions with capped additive valuations.}

\item {\bf Relaxing the computational assumptions.} Our computational
complexity results depend on the assumption that SAT does not have
polynomial-size circuits. Can this assumption be relaxed by proving
probabilistic versions of our VC machinery (see
Appendix~\ref{apx_ajtai})?

\item {\bf Characterizing truthfulness in auctions.} Can our inapproximability
results be made to hold for all truthful mechanisms? So far, despite
much work on this subject~\cite{R79, LMN, LMN2, LMNB, DS08, PSS08},
very little is known about characterizations of truthfulness in
combinatorial auctions (and in other multi-parameter environments).

\end{enumerate}

\bibliography{MIR}

\appendix

\section{Strengthening
Theorem~\ref{thm_lower_cover}}\label{apx_ajtai}

We shall now show how, if the valuation functions of the bidders are
slightly more expressive than capped additive valuations, one can
obtain a lower bound as in~\ref{thm_lower_cover} dependent on the
weaker computational assumption that NP is not contained in BPP.
This is achieved by using the probabilistic version of the
Sauer-Shelah Lemma presented by Ajtai~\cite{Ajtai98} to obtain a
\emph{probabilistic} polynomial-time reduction from an NP-hard
problem to the problem solved by the maximal-in-range mechanism. We
are currently unable to prove a similar result for capped additive
valuations.

We consider the class of double-capped additive valuations.
Informally, a bidder has a double-capped additive valuation if he
has an additive valuation, but also has some upper bound on how much
he is willing to spend on different subsets of items, as well as a
global upper bound on how much he is willing to spend overall.

\begin{definition} [double-capped additive valuations]
A valuation function $v$ is said to be a double-capped additive
valuation if there exists a partition of the set of items $[m]$ into
disjoint subsets $S_1,\ldots,S_r$ (that cover all items), an
additive valuation $a$, and real values $B,B_1,\ldots,B_r$, such
that for every bundle $S\subseteq [m]$, $v(S)=\min\{\Sigma_{t=1}^r
\min\{\Sigma_{j\in S_t} a(j),B_t\}, B\}$.
\end{definition}

We prove the following theorem:

\begin{thm}\label{thm_lower_cover-dbav}
For the allocate-all-items case, no polynomial-time maximal-in-range
mechanism obtains an approximation ratio of $2-\epsilon$ for
$2$-bidder combinatorial auctions with double-capped additive
valuations unless NP $\subseteq$ BPP.
\end{thm}

\begin{proof} [sketch]
Let $R$ be the range of a $(2-\epsilon)$-approximation
maximal-in-range mechanism. In the proof of
Theorem~\ref{thm_vc_covering} we show that $R$ must be of
exponential size. Let $R_1$ denote the subsets of items bidder $1$
is assigned in $R$. Because the allocate-all-items assumption holds
then, as explained in Lemma~\ref{lem_vc_cover}, $|R_1|$ is also of
exponential size. In Lemma~\ref{lem_vc_cover}, we used the
Sauer-Shelah Lemma to conclude that there must be a large set $E$
that is shattered in the traditional sense. Now, we make use of
Ajtai's probabilistic version of the Sauer-Shelah Lemma:

\begin{lemma}(\cite{Ajtai98})\label{lem_ajtai}
Let $Z$ be a family of subsets of a universe $R$ that is regular
(i.e., all subsets in $Z$ are of equal size) and $Q\geq
2^{|R|^{\alpha}}$ (for some $0<\alpha\leq 1$). There are integers
$q,l$ (where $|R|,q$ and $l$ are polynomially related) such that if
we randomly choose $q$ pairwise-disjoint subsets of $R$,
$Q_1,...,Q_q$, each of size $l$, then, w.h.p., for every function
$f:[q]\rightarrow\{0,1\}$ there is a subset $Z'\in Z$ for which
$|Z'\cap Q_j|=f(j)$ for all $j\in [q]$.
\end{lemma}

Using Ajtai's Lemma (with $Z=R_1$), we conclude that there must be
sets of items $Q_1,\ldots,Q_q$ as in the statement of the lemma.
Now, a reduction similar to that in Theorem~\ref{thm_lower_cover},
in which each item in the smaller auction is identified with
\emph{all} items in a specific $Q_s$ concludes the proof of the
Theorem.
\end{proof}

\end{document}